\documentclass[11pt,draft]{article}
\usepackage[letterpaper,hmargin=1in,vmargin=1.25in]{geometry}

\usepackage{url}

\usepackage{amsmath}
\usepackage{amssymb}

\usepackage{amsthm}

\theoremstyle{plain}
\newtheorem{theorem}{Theorem}[section]
\newtheorem{lemma}[theorem]{Lemma}
\newtheorem{corollary}[theorem]{Corollary}

\newtheorem{definition}[theorem]{Definition}

\newtheorem{replacements}{Replacements}

\newcommand{\abs}[1]{\left\lvert#1\right\rvert}

\DeclareMathOperator{\Dom}{dom}

\newcommand{\N}{\mathbb{N}}
\newcommand{\Q}{\mathbb{Q}}
\newcommand{\R}{\mathbb{R}}
\newcommand{\X}{\{0,1\}^*}
\newcommand{\K}{H}

\newcommand{\noi}{\noindent}

\title{\textbf{
Fixed point theorems on partial randomness%
\thanks{
A preliminary version of this paper appeared
in the Proceedings of
the Symposium on Logical Foundations of Computer Science 2009 (LFCS'09),
S.~Artemov and A.~Nerode (Eds.),
Lecture Notes in Computer Science, Springer-Verlag, Vol.5407,
pp.422--440,
January 3-6, 2009,
Deerfield Beach, Florida, USA.
}
}}

\author{
Kohtaro Tadaki\\
\\
Research and Development Initiative, Chuo University\\
CREST, JST\\
1--13--27 Kasuga, Bunkyo-ku, Tokyo 112-8551, Japan\\
E-mail: tadaki@kc.chuo-u.ac.jp\\
http://www2.odn.ne.jp/tadaki/
}

\date{}

\begin{document}

\maketitle

\begin{quotation}
\noi\textbf{Abstract.}
In our former work
[K.~Tadaki, Local Proceedings of CiE 2008, pp.~425--434, 2008],
we developed
a statistical mechanical interpretation of algorithmic information theory
by introducing the notion of thermodynamic quantities at temperature $T$,
such as free energy $F(T)$, energy $E(T)$,
and statistical mechanical entropy $S(T)$,
into
the theory.
These quantities are real functions of real argument $T>0$.
We then discovered that, in the interpretation,
the temperature $T$ equals to the partial randomness of the values of
all these thermodynamic quantities,
where the notion of partial randomness is
a stronger representation of
the compression rate
by
program-size complexity.
Furthermore,
we showed that this situation holds for the temperature itself
as a thermodynamic quantity.
Namely,
the computability of the value of
partition function $Z(T)$
gives a sufficient condition for
$T\in(0,1)$
to be a fixed point on partial randomness.
In this paper,
we show that
the computability of each of all the thermodynamic quantities above
gives the sufficient condition also.
Moreover, we show
that the computability of $F(T)$ gives
completely different fixed points from the computability of $Z(T)$.
\end{quotation}

\vspace{1mm}

\begin{quotation}
\noi\textit{Key words\/}:
algorithmic randomness,
fixed point theorem,
partial randomness,
Chaitin $\Omega$ number,
algorithmic information theory,
thermodynamic quantities
\end{quotation}

\begin{quotation}
\noi\textbf{AMS subject classifications (2000)}\;
68Q30, 26E40, 03D80, 82B30, 82B03
\end{quotation}
%
%
%
%
%
%

\section{Introduction}

Algorithmic information theory (AIT, for short) is a framework
for applying
information-theoretic and probabilistic ideas to recursive function theory.
One of the primary concepts of AIT is the \textit{program-size complexity}
(or \textit{Kolmogorov complexity}) $\K(s)$ of a finite binary string $s$,
which is defined as the length of the shortest binary
program
for the universal self-delimiting Turing machine $U$ to output $s$.
By the definition,
$\K(s)$
is thought to represent
the degree of randomness of a finite binary string $s$.
In particular,
the notion of program-size complexity plays a crucial role in
characterizing the \textit{randomness} of an infinite binary string,
or equivalently, a real number.

In
\cite{T08CiE}
we developed
a statistical mechanical interpretation of AIT.
In the development
we introduced especially the notion of
\textit{thermodynamic quantities},
such as partition function $Z(T)$, free energy $F(T)$, energy $E(T)$,
statistical mechanical entropy $S(T)$, and specific heat $C(T)$,
into AIT.
These
quantities are real numbers which depend
on \textit{temperature} $T$,
any positive real number.
We then proved that
if the temperature $T$ is a computable real number with $0<T<1$
then, for each of these thermodynamic quantities,
the partial randomness of its value equals to $T$,
where the notion of \textit{partial randomness} is
a stronger representation of the compression rate
by means of program-size complexity.
Thus,
the temperature $T$ plays a role as
the partial randomness of all the thermodynamic quantities
in the statistical mechanical interpretation of AIT.
%
In \cite{T08CiE}
we further showed that
the temperature $T$ plays a role as the partial randomness of
the temperature $T$ itself,
which is a thermodynamic quantity of itself.
Namely,
we proved \textit{the fixed point theorem on partial randomness},%
\footnote{
The fixed point theorem on partial randomness is called
a fixed point theorem on compression rate in \cite{T08CiE}.}
which states that, for every $T\in(0,1)$,
if the value of partition function $Z(T)$ at temperature $T$
is a computable real number,
then the partial randomness of $T$ equals to $T$,
and therefore the compression rate of $T$ equals to $T$,
i.e.,
$\lim_{n\to\infty}H(T_n)/n=T$,
where $T_n$ is the first $n$ bits of the base-two expansion of $T$.

In this paper,
we show that
a fixed point theorem
of the same form
as for $Z(T)$
holds also for each of free energy $F(T)$, energy $E(T)$, and
statistical mechanical entropy $S(T)$.
Moreover,
based on
the statistical mechanical relation $F(T)=-T\log_2 Z(T)$,
we show that the computability of $F(T)$ gives
completely different fixed points from the computability of $Z(T)$.

The paper is organized as follows.
We begin in Section \ref{preliminaries} with
some preliminaries to
AIT
and partial randomness.
In Section \ref{tcr},
we review the previous results \cite{T08CiE} on
the statistical mechanical interpretation of AIT
and the fixed point theorem by $Z(T)$,
which is given as Theorem~\ref{main} in the present paper.
Our main results;
the fixed point theorems
by $F(T)$, $E(T)$, and $S(T)$, are
presented in Section \ref{fpts},
and their proofs are completed in Section \ref{proof}.
In the last section,
we investigate some properties of the sufficient conditions
for $T$ to be a fixed point in the fixed point theorems.

\section{Preliminaries}
\label{preliminaries}

\subsection{Basic notation}
\label{basic notation}

We start with some notation about numbers and strings
which will be used in this paper.
$\N=\left\{0,1,2,3,\dotsc\right\}$ is the set of natural numbers,
and $\N^+$ is the set of positive integers.
$\Q$ is the set of rational numbers, and
$\R$ is the set of real numbers.
Let $f\colon S\to\R$ with $S\subset\R$.
We say that $f$ is \textit{increasing} (resp., \textit{decreasing})
if $f(x)<f(y)$ (resp., $f(x)>f(y)$) for all $x,y\in S$ with $x<y$.
We denote by $f'$ the derived function of $f$.

Normally, $o(n)$ denotes any
function $f\colon \N^+\to\R$ such
that $\lim_{n \to \infty}f(n)/n=0$.
On the other hand,
$O(1)$ denotes any
function $g\colon \N^+\to\R$ such that
there is $C\in\R$ with the property that
$\abs{g(n)}\le C$ for all $n\in\N^+$.

$\X=
\left\{
  \lambda,0,1,00,01,10,11,000,\dotsc
\right\}$
is the set of finite binary strings,
where $\lambda$ denotes the \textit{empty string}.
For any $s \in \X$, $\abs{s}$ is the \textit{length} of $s$.
A subset $S$ of $\X$ is called
\textit{prefix-free}
if no string in $S$ is a prefix of another string in $S$.
For any partial function $f$,
the domain of definition of $f$ is denoted by $\Dom f$.
We write ``r.e.'' instead of ``recursively enumerable.''

Let $\alpha$ be an arbitrary real number.
$\lfloor \alpha \rfloor$ is the greatest integer less than or equal to $\alpha$,
and $\lceil \alpha \rceil$ is the smallest integer greater than or equal to $\alpha$.
For any $n\in\N^+$,
we denote by $\alpha_n\in\X$
the first $n$ bits of the base-two expansion of
$\alpha - \lfloor \alpha \rfloor$ with infinitely many zeros.
For example,
in the case of $\alpha=5/8$,
$\alpha_6=101000$.

We say that a real number $\alpha$ is \textit{computable} if
there exists a total recursive function $f\colon\N^+\to\Q$ such that
$\abs{\alpha-f(n)} < 1/n$ for all $n\in\N^+$.
We say that $\alpha$ is \textit{left-computable} if
there exists a total recursive function $g\colon\N^+\to\Q$ such that
$g(n)\le\alpha$ for all $n\in\N^+$ and $\lim_{n\to\infty} g(n)=\alpha$.
On the other hand,
we say that
a real number
$\alpha$ is \textit{right-computable} if
$-\alpha$ is left-computable.
%
The following (i) and (ii) then hold:
\begin{enumerate}
  \item
    A real number
    $\alpha$ is computable if and only if
    $\alpha$ is both left-computable and right-computable.
  \item
    A real number
    $\alpha$ is right-computable if and only if
    the set $\{\,r\in\Q\mid \alpha<r\,\}$ is r.e.
\end{enumerate}
See e.g.~Weihrauch \cite{W00}
for the detail of the treatment of
the computability of real numbers.

\subsection{Algorithmic information theory}
\label{ait}

In the following
we concisely review some definitions and results of
algorithmic information theory
\cite{C75,C87a,C87b}.
A \textit{computer} is a partial recursive function
$C\colon \X\to \X$
such that
$\Dom C$ is a prefix-free set.
For each computer $C$ and each $s \in \X$,
$\K_C(s)$ is defined by
$\K_C(s) =
\min
\left\{\,
  \abs{p}\,\big|\;p \in \X\>\&\>C(p)=s
\,\right\}$
(may be $\infty$).
A computer $U$ is said to be \textit{optimal} if
for each computer $C$ there exists $d\in\N$
with the following property;
if $C(p)$ is defined, then there is a $p'$ for which
$U(p')=C(p)$ and $\abs{p'}\le\abs{p}+d$.
It is easy to see that there exists an optimal computer.
Note
that the class of optimal computers equals to
the class of functions which are computed
by \textit{universal self-delimiting Turing machines}
(see Chaitin \cite{C75} for the detail).
We choose a particular optimal computer $U$
as the standard one for use,
and define $\K(s)$ as $\K_U(s)$,
which is referred to as
the \textit{program-size complexity} of $s$,
the \textit{information content} of $s$, or
the \textit{Kolmogorov complexity} of $s$
\cite{G74,L74,C75}.
It follows that
for every computer $C$ there exists $d\in\N$ such that,
every $s\in\X$,
\begin{equation}
  H(s) \le H_C(s) + d. \label{eq: k}
\end{equation}
Based on this we can show
that there exists $c\in\N$ such that,
for every $s \neq \lambda$,
\begin{equation}
  H(s)\le\abs{s}+2\log_2\abs{s}+c. \label{eq: fas}
\end{equation}
An element of $\Dom U$ is called a \textit{program for $U$}.

\textit{Chaitin's halting probability} $\Omega$ is defined by
$\Omega=\sum_{p\in\Dom U}2^{-\abs{p}}$.
For any $\alpha\in\R$,
we say that $\alpha$ is \textit{weakly Chaitin random}
if there exists $c\in\N$ such that
$n-c\le \K(\alpha_n)$ for all $n\in\N^+$
\cite{C75,C87b}.
Then
Chaitin
\cite{C75} showed that $\Omega$ is weakly Chaitin random.
For any $\alpha\in\R$,
we say that $\alpha$ is
\textit{Chaitin random}
if $\lim_{n\to \infty} \K(\alpha_n)-n=\infty$ \cite{C75,C87b}.
It is then shown that,
for every $\alpha\in\R$,
$\alpha$ is weakly Chaitin random if and only if
$\alpha$ is Chaitin random
(see
Chaitin
\cite{C87b} for the proof and historical detail).
Thus $\Omega$ is Chaitin random.

\subsection{Partial randomness}
\label{partial}

In the works \cite{T99,T02},
we generalized the notion of
the randomness of a real number
so that \textit{the degree of the randomness},
which is often referred to as
\textit{the partial randomness} recently
\cite{CST06,RS05,CS06},
can be characterized by a real number $T$
with $0\le T\le 1$ as follows.


\begin{definition}[weak Chaitin $T$-randomness]
Let
$T\in\R$ with
$T\ge 0$.
For any $\alpha\in\R$,
we say that $\alpha$ is \textit{weakly Chaitin $T$-random} if
there exists $c\in\N$ such that
$Tn-c \le H(\alpha_n)$
for all $n\in\N^+$.
\qed
\end{definition}

\begin{definition}[$T$-compressibility]
Let
$T\in\R$ with
$T\ge 0$.
For any $\alpha\in\R$,
we say that $\alpha$ is \textit{$T$-compressible} if
$H(\alpha_n)\le Tn+o(n)$,
which is equivalent to
$\limsup_{n \to \infty}H(\alpha_n)/n\le T$.
\qed
\end{definition}

In the case of $T=1$,
the weak Chaitin $T$-randomness results in the weak Chaitin randomness.
For every $T\in[0,1]$ and every $\alpha\in\R$,
if $\alpha$ is weakly Chaitin $T$-random and $T$-compressible,
then
\begin{equation}\label{compression-rate}
  \lim_{n\to \infty} \frac{H(\alpha_n)}{n} = T.
\end{equation}
The
left-hand side of \eqref{compression-rate}
is referred to as the \textit{compression rate} of
a real number $\alpha$ in general.
Note, however, that \eqref{compression-rate}
does not necessarily imply that $\alpha$ is weakly Chaitin $T$-random.
Thus, the notion of partial randomness is
a stronger representation of compression rate.

\begin{definition}[Chaitin $T$-randomness, Tadaki \cite{T99,T02}]
  Let $T\in\R$ with $T\ge 0$.
  For any $\alpha \in\R$,
  we say that $\alpha$ is \textit{Chaitin $T$-random} if
  $\lim_{n\to \infty} H(\alpha_n)-Tn = \infty$.
  \qed
\end{definition}

In the case of $T=1$,
the Chaitin $T$-randomness results in the Chaitin randomness.
Obviously,
for every $T\in[0,1]$ and every $\alpha\in\R$,
if $\alpha$ is Chaitin $T$-random,
then $\alpha$ is weakly Chaitin $T$-random.
However,
in 2005 Reimann and Stephan \cite{RS05} showed that,
in the case of $T<1$,
the converse does not necessarily hold.
This contrasts with the
equivalence
between
the weak Chaitin randomness and the Chaitin randomness,
each of which corresponds to the case of $T=1$.
%
%
Recently, Kjos-Hanssen \cite{K08b}
showed
that the distinction between
the weak Chaitin $T$-randomness and the Chaitin $T$-randomness
has important applications to the research on
the notion of $T$-capacitability and its related notions
\cite{K08a,R08}.

\section{The previous results}
\label{tcr}

In this section,
we review some results of the statistical mechanical interpretation of
AIT,
developed by our former work \cite{T08CiE}.
We first introduce the notion of thermodynamic quantities into AIT
in the following manner.

In statistical mechanics,
the partition function $Z_{\mathrm{sm}}(T)$,
free energy $F_{\mathrm{sm}}(T)$,
energy $E_{\mathrm{sm}}(T)$,
entropy $S_{\mathrm{sm}}(T)$, and
specific heat $C_{\mathrm{sm}}(T)$
at temperature $T$
are given
as follows:
\begin{equation}\label{tdqsm}
\begin{split}
  Z_{\mathrm{sm}}(T)
  &=\sum_{x\in X}e^{-\frac{E_x}{k_{\mathrm{B}}T}},
  \hspace*{28mm}
  F_{\mathrm{sm}}(T)
  =-k_{\mathrm{B}}T\ln Z_{\mathrm{sm}}(T), \\
  E_{\mathrm{sm}}(T)
  &=\frac{1}{Z_{\mathrm{sm}}(T)}\sum_{x\in X}E_xe^{-\frac{E_x}{k_{\mathrm{B}}T}},
  \hspace*{9mm}
  S_{\mathrm{sm}}(T)
  =\frac{E_{\mathrm{sm}}(T)-F_{\mathrm{sm}}(T)}{T}, \\
  &\hspace*{25mm}
  C_{\mathrm{sm}}(T)
  =\frac{d}{dT}E_{\mathrm{sm}}(T),
\end{split}
\end{equation}
where $X$ is a complete set of energy eigenstates of
a quantum system
and $E_x$ is the energy of an energy eigenstate $x$.
The constant $k_{\mathrm{B}}$ is called the Boltzmann Constant,
and the $\ln$ denotes the natural logarithm.%
\footnote{
For the thermodynamic quantities in statistical mechanics,
see e.g.~Chapter 16 of \cite{C85}
and Chapter 2 of \cite{TKS92}.
To be precise,
the partition function is not a thermodynamic quantity
but a statistical mechanical quantity.
}

We introduce the notion of thermodynamic quantities into AIT
by performing Replacements~\ref{CS06} below
for the thermodynamic quantities \eqref{tdqsm} in statistical mechanics.

\begin{replacements}\label{CS06}\hfill
\begin{enumerate}
  \item Replace the complete set $X$ of energy eigenstates $x$
    by the set $\Dom U$ of all programs $p$ for $U$.
  \item Replace the energy $E_x$ of an energy eigenstate $x$
    by the length $\abs{p}$ of a program $p$.
  \item Set the Boltzmann Constant $k_{\mathrm{B}}$ to $1/\ln 2$.\qed
\end{enumerate}
\end{replacements}

For that purpose,
we first choose a particular recursive enumeration
$p_1,p_2,p_3,p_4, \dotsc$ of the infinite r.e.~set $\Dom U$
as the standard one for use throughout
the rest of this paper.%
\footnote{
Actually,
the enumeration $\{p_i\}$ can be chosen quite arbitrarily,
and the results of this paper are independent of the choice of $\{p_i\}$.
This is because
the sum $\sum_{i=1}^k 2^{-\abs{p_i}/T}$ and
$\sum_{i=1}^k \abs{p_i}2^{-\abs{p_i}/T}$
in Definition~\ref{tdqait}
are positive term series
and converge as $k\to\infty$ for every $T\in(0,1)$
(see Lemma~\ref{uniform} (i) below).
For the sake of convenience,
however,
we choose $\{p_i\}$ to be a recursive enumeration of $\Dom U$
in this paper.
}
Then,
motivated by the
formulae \eqref{tdqsm}
and taking into account Replacements~\ref{CS06},
we introduce the notion of thermodynamic quantities into AIT
as follows.%

\begin{definition}[thermodynamic quantities in AIT,
\cite{T08CiE}]\label{tdqait}
Let $T$ be any real number with $T>0$.
\begin{enumerate}
\item
  The \textit{partition function} $Z(T)$
  at temperature $T$
  is defined
  as $\lim_{k\to\infty} Z_k(T)$
  where
  \begin{equation*}
    Z_k(T)=\sum_{i=1}^k 2^{-\frac{\abs{p_i}}{T}}.
  \end{equation*}
\item
  The \textit{free energy} $F(T)$
  at temperature $T$
  is defined
  as $\lim_{k\to\infty} F_k(T)$
  where
  \begin{equation*}
    F_k(T)=-T\log_2 Z_k(T).
  \end{equation*}
\item
  The \textit{energy} $E(T)$
  at temperature $T$
  is defined
  as $\lim_{k\to\infty} E_k(T)$
  where
  \begin{equation*}
    E_k(T)
    =\frac{1}{Z_k(T)}\sum_{i=1}^k \abs{p_i}2^{-\frac{\abs{p_i}}{T}}.
  \end{equation*}
\item
  The \textit{statistical mechanical entropy} $S(T)$
  at temperature $T$
  is defined as $\lim_{k\to\infty} S_k(T)$
  where
  \begin{equation*}
    S_k(T)=\frac{E_k(T)-F_k(T)}{T}.
  \end{equation*}
\item
  The \textit{specific heat} $C(T)$
  at temperature $T$
  is defined
  as $\lim_{k\to\infty} C_k(T)$
  where $C_k(T)=E_k'(T)$,
  the derived function of $E_k(T)$.\qed
\end{enumerate}
\end{definition}


Note that $Z(1)=\Omega$ in particular.
Then
Theorems~\ref{cprpffe} and \ref{cpreesh} below
hold for these thermodynamic quantities in AIT.

\begin{theorem}[properties of $Z(T)$ and $F(T)$,
\cite{T99,T02,T08CiE}]\label{cprpffe}
Let $T\in\R$.
\begin{enumerate}
  \item If $0<T\le 1$ and $T$ is computable,
    then each of $Z(T)$ and $F(T)$ converges
    and is
    weakly Chaitin $T$-random and $T$-compressible.
  \item If $1<T$,
    then $Z(T)$ and $F(T)$ diverge to $\infty$ and $-\infty$,
    respectively.\qed
\end{enumerate}
\end{theorem}

\begin{theorem}[properties of $E(T)$, $S(T)$, and $C(T)$,
\cite{T08CiE}]\label{cpreesh}
Let $T\in\R$.
\begin{enumerate}
  \item If $0<T<1$ and $T$ is computable,
    then each of $E(T)$, $S(T)$, and $C(T)$ converges
    and is
    Chaitin $T$-random and $T$-compressible.
  \item If $1\le T$,
    then both $E(T)$ and $S(T)$ diverge to $\infty$.
    In the case of $T=1$, $C(T)$ diverges to $\infty$.%
\footnote{
It is still open whether $C(T)$ diverges or not
in the case of $T>1$.
}\qed
\end{enumerate}
\end{theorem}

The above
two
theorems show that
if $T$ is
a computable real number with
$T\in(0,1)$
then the temperature $T$ equals to the partial randomness
(and therefore the compression rate) of
the values of all the thermodynamic quantities
in Definition~\ref{tdqait}.
Note also that
the weak Chaitin $T$-randomness of
thermodynamic quantities for $T<1$
in Theorems~\ref{cprpffe}
is strengthen to the Chaitin $T$-randomness
in Theorem~\ref{cpreesh}
in exchange for the divergence of
thermodynamic quantities
at $T=1$.

In statistical mechanics or thermodynamics,
among all thermodynamic quantities
one of the most typical thermodynamic quantities is temperature
itself.
Inspired by this fact in physics and the above observation
on the role of the temperature $T$
in
the statistical mechanical interpretation of AIT,
the following question
thus arises naturally:
Can the partial randomness of the temperature $T$
equal to the temperature $T$ itself
in
the
statistical mechanical interpretation of AIT~?
This question is rather self-referential.
However,
we can answer it affirmatively
in the following form.

\begin{theorem}[fixed point theorem on partial randomness,
\cite{T08CiE}]\label{main}
For every $T\in(0,1)$,
if $Z(T)$ is computable,
then $T$ is weakly Chaitin $T$-random and $T$-compressible,
and therefore
\begin{equation*}
  \lim_{n\to\infty}\frac{H(T_n)}{n}=T.
\end{equation*}
\qed
\end{theorem}

%
%
%

Theorem~\ref{main} is just a fixed point theorem on partial randomness,
where the computability of the value $Z(T)$
gives a sufficient condition
for a real number $T\in(0,1)$ to be a fixed point on partial randomness.
%
Thus,
the above observation
that the temperature $T$ equals to the partial randomness of
the values of the thermodynamic quantities
in the statistical mechanical interpretation of AIT
is further confirmed.
In this paper,
we confirm this observation much further
by showing that
fixed point theorems of the same form as Theorem~\ref{main} hold also for
free energy $F(T)$, energy $E(T)$,
and statistical mechanical entropy $S(T)$.
For completeness,
we include the proof of Theorem~\ref{main}
in Appendix~\ref{proof-main}.
%


\section{The main results}
\label{fpts}

The following three theorems are the main results of this paper.

\begin{theorem}[fixed point theorem by free energy]\label{fpt-free_energy}
For every $T\in(0,1)$,
if $F(T)$ is computable
then $T$ is weakly Chaitin $T$-random and $T$-compressible.\qed
\end{theorem}

\begin{theorem}[fixed point theorem by energy]\label{fpt-energy}
For every $T\in(0,1)$,
if $E(T)$ is computable
then $T$ is Chaitin $T$-random and $T$-compressible.\qed
\end{theorem}

\begin{theorem}[fixed point theorem by
statistical mechanical
entropy]\label{fpt-entropy}
For every $T\in(0,1)$,
if $S(T)$ is computable
then $T$ is Chaitin $T$-random and $T$-compressible.\qed
\end{theorem}

First, note that
the weak Chaitin $T$-randomness of $T$ in Theorems~\ref{main}
is strengthen to the Chaitin $T$-randomness of $T$
in Theorems~\ref{fpt-energy} and \ref{fpt-entropy}.

The proof of Theorem~\ref{fpt-free_energy} uses
Theorems~\ref{gfpwCTr}, \ref{gfpTc1}, and \ref{gfpTc2} below.
On the other hand,
the proofs of Theorems~\ref{fpt-energy} and \ref{fpt-entropy} use
Theorems~\ref{gfpCTr}, \ref{gfpTc1}, and \ref{gfpTc2} below.
All these
proofs
also use
Theorem~\ref{tdr} below,
where the thermodynamic relations in statistical mechanics
are recovered by the thermodynamic quantities of AIT.
%
We
complete
the proofs of
Theorems~\ref{fpt-free_energy},
\ref{fpt-energy},
and \ref{fpt-entropy}
in the next section.
Compared with the proof of Theorem~\ref{fpt-free_energy},
the proofs of Theorems~\ref{fpt-energy} and \ref{fpt-entropy}
are more delicate.

\begin{theorem}[thermodynamic relations]\label{tdr}\hfill
\begin{enumerate}
  \item $F_{k}'(T)=-S_{k}(T)$, $E_{k}'(T)=C_{k}(T)$,
    and $S_{k}'(T)=C_{k}(T)/T$
    for every $k\in\N^+$ and every $T\in(0,1)$.
  \item $F'(T)=-S(T)$, $E'(T)=C(T)$, and $S'(T)=C(T)/T$
    for every $T\in(0,1)$.
  \item $S_{k}(T),C_{k}(T)\ge 0$
    for every $k\in\N^+$ and every $T\in(0,1)$.
    There exists $k_0\in\N^+$ such that,
    for every $k\ge k_0$ and every $T\in(0,1)$,
    $S_{k}(T),C_{k}(T)>0$.
    Moreover, $S(T),C(T)>0$ for every $T\in(0,1)$.\qed
\end{enumerate}
\end{theorem}

The proof of Theorem~\ref{tdr} uses Lemma~\ref{uniform} below.
For each $T\in(0,1)$,
we define $W(T)$ and $Y(T)$
as $\lim_{k\to\infty} W_{k}(T)$ and $\lim_{k\to\infty} Y_{k}(T)$,
respectively,
where
$W_{k}(T)=\sum_{i=1}^k \abs{p_i}2^{-\frac{\abs{p_i}}{T}}$
and
$Y_{k}(T)=\sum_{i=1}^k \abs{p_i}^2 2^{-\frac{\abs{p_i}}{T}}$.

\begin{lemma}\label{uniform}\hfill
\begin{enumerate}
  \item For every $T\in(0,1)$,
    the limit values $Z(T)$, $W(T)$, and $Y(T)$ exist,
    and are positive real numbers.
  \item The sequence $\{Z_{k}(T)\}_k$ of functions of $T$
    is uniformly convergent on $(0,1)$
    in the wider sense.
    The same holds for
    the sequences $\{W_{k}(T)\}_k$ and $\{Y_{k}(T)\}_k$.
  \item The function $Z(T)$ of $T$ is continuous on $(0,1)$.
    The same holds for the functions $W(T)$ and $Y(T)$.
\end{enumerate}
\end{lemma}

\begin{proof}
(i)
Suppose that
$T$ is an arbitrary real number with $T\in(0,1)$.

First,
we show that $Y_{k}(T)$ converges as $k\to\infty$.
Since $T<1$,
there is $l_0\in\N^+$ such that
\begin{equation*}
  \frac{1}{T}-2\frac{\log_2 l}{l}\ge 1
\end{equation*}
for all $l\ge l_0$.
Then,
since $\lim_{k\to\infty}\abs{p_{k}}=\infty$,
there is $k_0\in\N^+$ such that $\abs{p_i}\ge l_0$ for all $i>k_0$.
Thus, we see that,
for each $i>k_0$,
\begin{equation*}
  \abs{p_i}^2 2^{-\frac{\abs{p_i}}{T}}
  =2^{-\left(\frac{1}{T}-2\frac{\log_2 \abs{p_i}}{\abs{p_i}}\right)
   \abs{p_i}}
  \le 2^{-\abs{p_i}}.
\end{equation*}
Hence, for each $k>k_0$,
\begin{equation*}
  Y_{k}(T)-Y_{k_0}(T)
  =\sum_{i=k_0+1}^{k} \abs{p_i}^2 2^{-\frac{\abs{p_i}}{T}}
  \le \sum_{i=k_0+1}^{k} 2^{-\abs{p_i}}<\Omega=Z(1).
\end{equation*}
Therefore,
since $\{Y_{k}(T)\}_k$ is
an increasing sequence of positive real numbers bounded to the above,
it converges to a positive real number as $k\to\infty$,
as desired.

Note that $0<Z_{k}(T)\le W_{k}(T)\le Y_{k}(T)$ for every $k\in\N^+$,
and the sequences $\{Z_{k}(T)\}_k$ and $\{W_{k}(T)\}_k$
of positive real numbers are increasing.
It follows that
$Z_{k}(T)$ and $W_{k}(T)$ converge to a positive real number
as $k\to\infty$.

(ii)
Note that,
for every $k\in\N^+$ and every $t,T\in(0,1)$ with $t\le T$, 
\begin{equation*}
  0
  <Z(t)-Z_{k}(t)
  =\sum_{i=k+1}^{\infty} 2^{-\frac{\abs{p_i}}{t}}
  \le\sum_{i=k+1}^{\infty} 2^{-\frac{\abs{p_i}}{T}}
  =Z(T)-Z_{k}(T).
\end{equation*}
It follows that
the sequence $\{Z_{k}(T)\}_k$ of functions of $T$
is uniformly convergent on $(0,1)$ in the wider sense.
In the same manner,
we can show that
the sequences $\{W_{k}(T)\}_k$ and $\{Y_{k}(T)\}_k$
are uniformly convergent on $(0,1)$ in the wider sense.

(iii)
Note that,
for each $k\in\N^+$,
the mapping $(0,1)\ni T\mapsto Z_{k}(T)$ is
a continuous function.
It follows from Lemma~\ref{uniform} (ii) that
the function $Z(T)$ of $T$ is continuous on $(0,1)$.
In the same manner,
we can show that
the functions $W(T)$ and $Y(T)$ of $T$ are continuous on $(0,1)$.
\end{proof}


The proof of Theorem~\ref{tdr} is then given as follows.

\begin{proof}[Proof of Theorem~\ref{tdr}]
(i)
First,
from Definition~\ref{tdqait}
we see that,
for every $k\in\N^+$ and every $T\in(0,1)$,
\begin{align}
  F_{k}(T)
  &=
  -T\log_2 Z_{k}(T),\label{FkT}\\
  E_{k}(T)
  &=
  \frac{W_{k}(T)}{Z_{k}(T)},\label{EkT}\\
  S_{k}(T)
  &=
  \frac{W_{k}(T)}{TZ_{k}(T)}+\log_2 Z_{k}(T),\label{SkT}\\
  Z_{k}'(T)
  &=
  \frac{\ln 2}{T^2}W_{k}(T),\label{Z'kT}\\
  W_{k}'(T)
  &=
  \frac{\ln 2}{T^2}Y_{k}(T).\label{W'kT}
\end{align}
Thus, by straightforward differentiation,
we can check that the relations of Theorem~\ref{tdr} (i) hold.
For example,
it follows from \eqref{SkT} and \eqref{EkT} that,
for every $k\in\N^+$ and every $T\in(0,1)$,
\begin{equation*}
  S_{k}'(T)
  =
  \frac{1}{T}E_{k}'(T)-\frac{1}{T^2}\frac{W_{k}(T)}{Z_{k}(T)}
  +\frac{1}{\ln 2}\frac{Z_{k}'(T)}{Z_{k}(T)}.
\end{equation*}
Using the definition $C_k(T)=E_k'(T)$
and the equation \eqref{Z'kT}
we see that,
for every $k\in\N^+$ and every $T\in(0,1)$,
$S_{k}'(T)=C_{k}(T)/T$.

(ii)
From \eqref{EkT}, \eqref{Z'kT}, \eqref{W'kT},
and the definition $C_k(T)=E_k'(T)$,
we see that,
for every $k\in\N^+$ and every $T\in(0,1)$,
\begin{equation}\label{CkT}
  C_{k}(T)
  =
  \frac{\ln 2}{T^2}
  \left\{
    \frac{Y_{k}(T)}{Z_{k}(T)}-\left(\frac{W_{k}(T)}{Z_{k}(T)}\right)^2
  \right\}.
\end{equation}
Thus,
using Lemma~\ref{uniform} (i) above and
the equations \eqref{FkT}, \eqref{EkT}, \eqref{SkT}, and \eqref{CkT},
we can first see that
the limit values $F(T)$, $E(T)$, $S(T)$, and $C(T)$ exist
for every $T\in(0,1)$.    
Using Lemma~\ref{uniform}
in whole
and the equations \eqref{SkT} and \eqref{CkT},
we can
next
check that
the sequences
$\{-S_{k}(T)\}_k$, $\{C_{k}(T)\}_k$ and $\{C_{k}(T)/T\}_k$
of functions of $T$
are uniformly convergent on $(0,1)$ in the wider sense.
Hence,
Theorem~\ref{tdr} (ii) follows immediately from Theorem~\ref{tdr} (i).

(iii)
From \eqref{SkT} 
we see that,
for every $k\in\N^+$ and every $T\in(0,1)$,
\begin{equation*}
  S_{k}(T)
  =
  -\sum_{i=1}^k
  \frac{2^{-\frac{\abs{p_i}}{T}}}{Z_{k}(T)}
  \log_2\frac{2^{-\frac{\abs{p_i}}{T}}}{Z_{k}(T)}.
\end{equation*}
Thus,
$S_{k}(T)\ge 0$
for every $k\in\N^+$.
We also see that,
for every $k\ge 2$ and every $T\in(0,1)$,
\begin{equation*}
  S_{k}(T)
  \ge
  -
  \frac{2^{-\frac{\abs{p_1}}{T}}}{Z_{k}(T)}
  \log_2\frac{2^{-\frac{\abs{p_1}}{T}}}{Z_{k}(T)}>0.
\end{equation*}
Hence, for every $T\in(0,1)$,
\begin{equation*}
  S(T)
  \ge
  -
  \frac{2^{-\frac{\abs{p_1}}{T}}}{Z(T)}
  \log_2\frac{2^{-\frac{\abs{p_1}}{T}}}{Z(T)}>0.
\end{equation*}

On the other hand,
from \eqref{CkT} 
we see that,
for every $k\in\N^+$ and every $T\in(0,1)$,
\begin{equation}\label{CkTv}
  C_{k}(T)
  =
  \frac{\ln 2}{T^2}
  \sum_{i=1}^k
  \{\abs{p_i}-E_{k}(T)\}^2\frac{2^{-\frac{\abs{p_i}}{T}}}{Z_{k}(T)}.
\end{equation}
Thus,
$C_{k}(T)\ge 0$
for every $k\in\N^+$ and every $T\in(0,1)$.
We note that
there exists $l\in\N^+$ such that
$\abs{p_{l}}\le \abs{p_{i}}$ for every $i\in\N^+$.
It is then easy to see that there exists $k_0\in\N^+$ such that,
for every $k\ge k_0$ and every $T\in(0,1)$,
$\abs{p_{l}}<E_{k}(T)$.
This is because there exists $i\in\N^+$ such that $\abs{p_{l}}<\abs{p_{i}}$.
Thus, by \eqref{CkTv} we see that,
for every $k\ge\max\{l,k_0\}$ and every $T\in(0,1)$,
\begin{equation}\label{CkTl}
  C_{k}(T)
  \ge
  \frac{\ln 2}{T^2}
  \{\abs{p_l}-E_{k}(T)\}^2\frac{2^{-\frac{\abs{p_l}}{T}}}{Z_{k}(T)}
  >0.
\end{equation}
It is also easy to see that $\abs{p_{l}}<E(T)$ for every $T\in(0,1)$.
It follows from \eqref{CkTl} that $C(T)>0$ for every $T\in(0,1)$.
\end{proof}

\begin{theorem}\label{gfpwCTr}
Let $f\colon (0,1)\to\R$.
Suppose that $f$ is increasing and
there exists $g\colon (0,1)\times\N^+\to\R$ which satisfies
the following four conditions:
\begin{enumerate}
  \item
    $\lim_{k\to\infty} g(T,k)=f(T)$ for every $T\in(0,1)$.
  \item
    $\{(q,r,k)\in\Q\times(\Q\cap(0,1))\times\N^+\mid q<g(r,k)\}$
    is an r.e.~set.
  \item
    For every $T\in(0,1)$,
    there exist $a\in\N$, $k_0\in\N^+$, and $t\in(T,1)$ 
    such that, for every $k\ge k_0$ and every $x\in(T,t)$,
    $g(x,k)-g(T,k)\le 2^a(x-T)$.
  \item
    For every $T\in(0,1)$,
    there exist $b\in\N$ and $k_1\in\N^+$
    such that, for every $k\ge k_1$,
    \begin{equation*}
      2^{-\frac{\abs{p_{k+1}}}{T}-b}\le
      g(T,k+1)-g(T,k).
    \end{equation*}
\end{enumerate}
Then, for every $T\in(0,1)$,
if $f(T)$ is right-computable
then $T$ is weakly Chaitin $T$-random.
\end{theorem}

\begin{proof}
The proof of Theorem~\ref{gfpwCTr} is obtained
by slightly simplifying the proof of Theorem~\ref{gfpCTr} below.
\end{proof}

\begin{theorem}\label{gfpCTr}
Let $f\colon (0,1)\to\R$.
Suppose that $f$ is increasing and
there exists $g\colon (0,1)\times\N^+\to\R$ which satisfies
the following four conditions:
\begin{enumerate}
  \item
    $\lim_{k\to\infty} g(T,k)=f(T)$ for every $T\in(0,1)$.
  \item
    $\{(q,r,k)\in\Q\times(\Q\cap(0,1))\times\N^+\mid q<g(r,k)\}$
    is an r.e.~set.
  \item
    For every $T\in(0,1)$,
    there exist $a\in\N$, $k_0\in\N^+$, and $t\in(T,1)$ 
    such that, for every $k\ge k_0$ and every $x\in(T,t)$,
    $g(x,k)-g(T,k)\le 2^a(x-T)$.
  \item
    For every $T\in(0,1)$,
    there exist $b\in\N$, $c\in\N^+$, and $k_1\in\N^+$
    such that, for every $k\ge k_1$,
    \begin{equation*}
      |p_{k+1}|^c2^{-\frac{\abs{p_{k+1}}}{T}-b}\le
      g(T,k+1)-g(T,k).
    \end{equation*}
\end{enumerate}
Then, for every $T\in(0,1)$,
if $f(T)$ is right-computable
then $T$ is Chaitin $T$-random.
\end{theorem}

\begin{proof}
Suppose that $T\in(0,1)$ and $f(T)$ is right-computable.
Then there exists a total recursive function
$h\colon\N^+\to\Q$
such that
$f(T)\le h(m)$ for all $m\in\N^+$ and
$\lim_{m\to\infty} h(m)=f(T)$.

Since the condition (iii) holds for $g$,
there exist $a\in\N$, $k_0\in\N^+$, and $t\in(T,1)$ such that
\begin{equation}\label{mvtu}
  g(x,k)-g(T,k)\le 2^a(x-T)
\end{equation}
for every $k\ge k_0$ and every $x\in(T,t)$.
We choose any one $n_0\in\N^+$ such that
$0.T_n+2^{-n}<t$ for all $n\ge n_0$.
Such $n_0$ exists since $T<t$ and $\lim_{n\to\infty} 0.T_n+2^{-n}=T$.
Since $T_n$ is the first $n$ bits of
the base-two expansion of $T$ with infinitely many zeros,
we further see that $T<0.T_n+2^{-n}<t$ for all $n\ge n_0$.

On the other hand,
since the condition (iv) holds for $g$,
there exist $b\in\N$, $c\in\N^+$, and $k_1\in\N^+$ such that
\begin{equation*}
  |p_{k+1}|^c2^{-\frac{\abs{p_{k+1}}}{T}-b}\le
  g(T,k+1)-g(T,k)
\end{equation*}
for every $k\ge k_1$.
Without loss of generality,
we can assume that $k_1=k_0$.
Thus,
since $g(T,k)$ is increasing on $k$ with $k\ge k_0$
and the condition (i) holds,
\begin{equation}\label{ts}
  \abs{p_{i}}^c2^{-\frac{\abs{p_{i}}}{T}-b}<f(T)-g(T,k)
\end{equation}
if $i>k\ge k_0$.

Now,
given $T_n$ with $n\ge n_0$,
one can effectively find $k_e,m_e\in\N^+$ such that
$k_e\ge k_0$ and $h(m_e)<g(0.T_n+2^{-n},k_e)$.
This is possible
because
$f(T)<f(0.T_n+2^{-n})$,
$\lim_{k\to\infty}g(0.T_n+2^{-n},k)=f(0.T_n+2^{-n})$,
and the condition (ii) holds for $g$.
It follows from $f(T)\le h(m_e)$
and \eqref{mvtu} that
$f(T)-g(T,k_e)<g(0.T_n+2^{-n},k_e)-g(T,k_e)\le 2^{a-n}$.
It follows from \eqref{ts} that,
for every $i>k_e$,
$\abs{p_{i}}^c2^{-\frac{\abs{p_{i}}}{T}-b}<2^{a-n}$
and therefore
$cT\log_2\abs{p_i}-(a+b)T<\abs{p_i}-Tn$.
Thus,
by calculating the set $\{\,U(p_i)\mid i\le k_e\,\}$
and picking any one finite binary string $s$ which is not in this set,
one can then obtain $s\in\X$ such that
$cT\log_2H(s)-(a+b)T<H(s)-Tn$.

Hence,
there exists a partial recursive function
$\Psi\colon\X\to\X$
such that
\begin{equation*}
  cT\log_2H(\Psi(T_n))-(a+b)T<H(\Psi(T_n))-Tn
\end{equation*}
for all $n\ge n_0$.
Applying this inequality to itself, we have
$cT\log_2n<H(\Psi(T_n))-Tn+O(1)$,
for all $n\in\N^+$.
On the other hand,
using \eqref{eq: k}, there is $c_\Psi\in\N$ such that
$H(\Psi(T_n))\le H(T_n)+c_\Psi$
for all $n\ge n_0$.
It follows that $cT\log_2n<H(T_n)-Tn+O(1)$.
Hence, $T$ is Chaitin $T$-random.
\end{proof}

\begin{theorem}\label{gfpTc1}
Let $f\colon (0,1)\to\R$.
Suppose that $f$ is increasing and
there exists $g\colon (0,1)\times\N^+\to\R$
which satisfies the following three conditions:
\begin{enumerate}
  \item For every $T\in(0,1)$,
    $\lim_{k\to\infty} g(T,k)=f(T)$.
  \item For every $T_1,T_2\in(0,1)$ with $T_1<T_2$,
    there exists $k_0\in\N^+$ such that,
    for every $k\ge k_0$ and every $x\in[T_1,T_2]$,
    $g(x,k)\le f(x)$.
  \item The set
    $\{(q,r,k)\in\Q\times(\Q\cap(0,1))\times\N^+\mid q<g(r,k)\}$
    is r.e.
\end{enumerate}
Then, for every $T\in(0,1)$,
if $f(T)$ is right-computable
then $T$ is also right-computable.
\end{theorem}

\begin{proof}
Suppose that $T\in(0,1)$.
We choose any $t_1,t_2\in\Q$ with $0<t_1<T<t_2<1$.
Then,
since the condition (ii) holds for $g$,
there exists $k_0\in\N^+$ such that $g(x,k)\le f(x)$
for every $k\ge k_0$ and every $x\in[t_1,t_2]$.
Suppose further that $f(T)$ is right-computable.
Then there exists a total recursive function
$h\colon\N^+\to\Q$
such that
$f(T)\le h(m)$ for all $m\in\N^+$ and
$\lim_{m\to\infty} h(m)=f(T)$.
Thus,
since $f$ is increasing and the condition (i) holds for $g$,
we see that,
for every $r\in\Q\cap[t_1,t_2]$,
$T<r$ if and only if
$\exists\,m\;\exists\,k\ge k_0\;h(m)<g(r,k)$.
Since the condition (iii) holds for $g$,
the set
$\{\,r\in\Q\cap[t_1,t_2]\mid
\exists\,m\;\exists\,k\ge k_0\;h(m)<g(r,k)\,\}$
is r.e.~and therefore
the set $\{\,r\in\Q\cap[t_1,t_2]\mid T<r\,\}$ is r.e.
It follows from $T\in(t_1,t_2)$ that $T$ is right-computable.
\end{proof}

\begin{theorem}\label{gfpTc2}
Let $f\colon (0,1)\to\R$.
Suppose that there exists $g\colon (0,1)\times\N^+\to\R$
which satisfies the following six conditions:
\begin{enumerate}
  \item For every $T\in(0,1)$, $\lim_{k\to\infty} g(T,k)=f(T)$.
  \item For every $T\in(0,1)$,
    there exists $k_0\in\N^+$ such that, for every $k\ge k_0$,
    $g(T,k)<f(T)$.
  \item For every $T\in(0,1)$,
    there exist $a\in\N$, $k_1\in\N^+$, and $t\in(T,1)$
    such that, for every $k\ge k_1$ and every $x\in(T,t)$,
    $g(x,k)-g(T,k)\ge 2^{-a}(x-T)$.
  \item For every $T\in(0,1)$,
    there exist $b\in\N$, $c\in\N$, and $k_2\in\N^+$
    such that, for every $k\ge k_2$,
    \begin{equation*}
      g(T,k+1)-g(T,k)\le
      \abs{p_{k+1}}^{b}
      2^{-\abs{p_{k+1}}/T+c}.
    \end{equation*}
  \item For each $k\in\N^+$,
    the mapping $(0,1)\ni T\mapsto g(T,k)$ is
    a
    continuous function.
  \item
    $\{(q,r,k)\in\Q\times(\Q\cap(0,1))\times\N^+\mid q>g(r,k)\}$
    is an r.e.~set.
\end{enumerate}
Then, for every $T\in(0,1)$,
if $f(T)$ is left-computable and $T$ is right-computable,
then $T$ is $T$-compressible.
\end{theorem}

\begin{proof}
Suppose that $T\in(0,1)$.
Since the condition (ii) holds for $g$,
there exists $k_0\in\N^+$ such that
\begin{equation}\label{gltf}
  g(T,k)<f(T)
\end{equation}
for every $k\ge k_0$.    
Since the condition (iii) holds for $g$,
there exist $a\in\N$, $k_1\in\N^+$, and $t\in(T,1)$ such that
\begin{equation}\label{mvtl}
  g(x,k)-g(T,k)\ge 2^{-a}(x-T)
\end{equation}
for every $k\ge k_1$ and every $x\in(T,t)$.
Since the condition (iv) holds for $g$,
there exist $b\in\N$, $c\in\N$, and $k_2\in\N^+$
such that
\begin{equation}\label{mvtu2}
  g(T,k+1)-g(T,k)\le
  \abs{p_{k+1}}^{b}2^{-\abs{p_{k+1}}/T+c}
\end{equation}
for every $k\ge k_2$.
Without loss of generality,
we can assume that $k_0=k_1=k_2$.

Suppose further that
$T$ is right-computable and $f(T)$ is left-computable. 
Then there exists a total recursive function $A\colon\N^+\to\Q$ such that
$T< A(l)< t$ for all $l\in\N^+$ and $\lim_{l\to\infty} A(l)=T$,
and there exists a total recursive function $B\colon\N^+\to\Q$ such that
$B(m)\le f(T)$ for all $m\in\N^+$ and $\lim_{m\to\infty} B(m)=f(T)$.

Let $u$ be an arbitrary computable real number with $T<u<1$,
and let
$\beta=\sum_{i=1}^{\infty} \abs{p_{i}}^{b}2^{-\abs{p_{i}}/u}$.
Note that this limit exists and is weakly Chaitin $u$-random
(see
Theorem 3.2 (a) of \cite{T02} and
Theorem 3 (i) of \cite{T08CiE}).
Thus, the base-two expansion of $\beta$ contains
infinitely many zeros and
infinitely many ones.

Given $n$ and $\beta_{\lceil Tn/u\rceil}$
(i.e., the first $\lceil Tn/u\rceil$ bits of the base-two expansion of
$\beta - \lfloor \beta \rfloor$),
one can effectively find
$k_e\in\N^+$ such that $k_e\ge k_0$ and
\begin{equation*}
  0.\beta_{\lceil Tn/u\rceil}+\lfloor \beta \rfloor
  <\sum_{i=1}^{k_e} \abs{p_{i}}^{b}2^{-\frac{\abs{p_i}}{u}}.
\end{equation*}
This is possible since
$0.\beta_{\lceil Tn/u\rceil}+\lfloor \beta \rfloor<\beta$ and
$\lim_{k\to\infty}\sum_{i=1}^{k} \abs{p_{i}}^{b}2^{-\abs{p_i}/u}
=\beta$.
Since
$\beta-(0.\beta_{\lceil Tn/u\rceil}+\lfloor \beta \rfloor)
\le 2^{-\lceil Tn/u\rceil}\le 2^{-Tn/u}$,
it is then shown that
\begin{equation*}
  \sum_{i=k_e+1}^{\infty} \abs{p_{i}}^{b}2^{-\frac{\abs{p_i}}{u}}
  =\beta-\sum_{i=1}^{k_e} \abs{p_{i}}^{b}2^{-\frac{\abs{p_i}}{u}}
  <2^{-Tn/u}.
\end{equation*}
Raising both
ends
of this inequality to the power $u/T$ and
using the inequality $x^z+y^z\le (x+y)^z$
for real numbers $x,y>0$ and $z\ge 1$,
we have
\begin{equation*}
  \sum_{i=k_e+1}^{\infty} \abs{p_{i}}^{b}2^{-\frac{\abs{p_i}}{T}}
  \le
  \sum_{i=k_e+1}^{\infty}
  \abs{p_{i}}^{\frac{bu}{T}}2^{-\frac{\abs{p_i}}{T}}
  <2^{-n}.
\end{equation*}
Using \eqref{mvtu2} and the condition (i),
it follows that
\begin{equation}\label{tcnew}
  f(T)-g(T,k_e)
  <\sum_{i=k_e+1}^{\infty} \abs{p_{i}}^{b}2^{-\frac{\abs{p_i}}{T}+c}
  <2^{c-n}.
\end{equation}
On the other hand,
since the condition (v) holds for $g$,
$g(T,k_e)=\lim_{l\to\infty} g(A(l),k_e)$.
Obviously, $g(T,k_e)<f(T)$ by \eqref{gltf}.
Thus, since the condition (vi) holds for $g$,
one can then effectively find $l_e, m_e\in\N^+$ such that
$g(A(l_e),k_e)<B(m_e)$.
It follows from \eqref{tcnew} and \eqref{mvtl} that
\begin{equation*}
  2^{c-n}>f(T)-g(T,k_e)\ge B(m_e)-g(T,k_e)
  >g(A(l_e),k_e)-g(T,k_e)\ge 2^{-a}(A(l_e)-T).
\end{equation*}
Thus, $0<A(l_e)-T<2^{a+c-n}$.
Let $r_n$ be the first $n$ bits of the base-two expansion of
the rational number $A(l_e)$ with infinitely many zeros.
Then $\abs{\,A(l_e)-0.r_n\,}<2^{-n}$.
It follows from $\abs{\,T-0.T_n\,}<2^{-n}$ that
$\abs{\,0.T_n-0.r_n\,}<(2^{a+c}+2)2^{-n}$.
Hence,
$T_n
=r_n,\,r_n\pm 1,\,r_n\pm 2,\,\dots,\,r_n\pm (2^{a+c}+1)$,
where $T_n$ and $r_n$ are regarded as a dyadic integer.
Thus, there are still $2^{a+c+1}+3$ possibilities of $T_n$,
so that one needs only $a+c+3$ bits more in order to determine $T_n$.

Thus, there exists a partial recursive function
$\Phi\colon \N^+\times\X\times\X\to\X$
such that,
for every $n\in\N^+$, there exists $s\in\X$
with the properties that
$\abs{s}=a+c+3$ and $\Phi(n,\beta_{\lceil Tn/u\rceil},s)=T_n$.
It follows from \eqref{eq: fas} that
$H(T_n)
\le |\beta_{\lceil Tn/u\rceil}|+o(n)
\le Tn/u+o(n)$,
which implies that $T$ is $T/u$-compressible.
Since $u$ is an arbitrary computable real number with $T<u<1$,
it follows that $T$ is $T$-compressible.
\end{proof}

\section{The proofs of the main results}
\label{proof}

In this section
we complete the proofs of
our main results;
Theorems~\ref{fpt-free_energy},
\ref{fpt-energy},
and
\ref{fpt-entropy}.

\subsection{The proof of Theorem~\ref{fpt-free_energy}}
\label{proof-fpt-free_energy}

We
first
complete the proof of
Theorem~\ref{fpt-free_energy},
based on Theorems~\ref{tdr}, \ref{gfpwCTr}, \ref{gfpTc1}, and \ref{gfpTc2},
as follows.

Let $f\colon (0,1)\to\R$ with $f(T)=-F(T)$,
and let $g\colon (0,1)\times\N^+\to\R$ with $g(T,k)=-F_{k}(T)$.
First, it follows from Theorem~\ref{tdr} (ii) and (iii)
that $f$ is increasing.

Obviously, $\lim_{k\to\infty} g(T,k)=f(T)$ for every $T\in(0,1)$.
Using the mean value theorem we see that,
for every $T\in(0,1)$ and every $k\in\N^+$,
\begin{equation}\label{lnZ}
  \frac{2^{-\frac{\abs{p_{k+1}}}{T}}}{Z_{k+1}(T)}
  <
  \ln Z_{k+1}(T)-\ln Z_{k}(T)
  <
  \frac{2^{-\frac{\abs{p_{k+1}}}{T}}}{Z_{k}(T)}.
\end{equation}
It follows that,
for every $T\in(0,1)$ and every $k\in\N^+$,
$g(T,k)<g(T,k+1)$ and therefore $g(T,k)<f(T)$.
At this point,
the conditions (i) and (ii) of Theorem~\ref{gfpwCTr},
all conditions of Theorem~\ref{gfpTc1},
and the conditions (i), (ii), (v), and (vi) of Theorem~\ref{gfpTc2}
hold for $g$.

Using \eqref{lnZ} we see that,
for every $T\in(0,1)$ and every $k\in\N^+$,
\begin{equation*}
  \frac{T2^{-\frac{\abs{p_{k+1}}}{T}}}{Z_{k+1}(T)\ln 2}
  <
  g(T,k+1)-g(T,k)
  <
  \frac{T2^{-\frac{\abs{p_{k+1}}}{T}}}{Z_{k}(T)\ln 2}.
\end{equation*}
Thus,
the condition (iv) of Theorem~\ref{gfpwCTr}
and the condition (iv) of Theorem~\ref{gfpTc2}
hold for $g$.

Using the mean value theorem and Theorem~\ref{tdr} (i) and (iii),
we see that
\begin{equation}\label{diff-g-T}
  S_{k}(T)(x-T)\le g(x,k)-g(T,k)\le S_{k}(t)(x-T)
\end{equation}
for every $k\in\N^+$ and every $T,x,t\in(0,1)$ with $T<x<t$.
On the other hand, we see that,
for every $k\in\N^+$ and every $T\in(0,1)$,
\begin{equation*}
  E_{k+1}(T)-E_{k}(T)
  =
  \frac{Z_{k}(T)\abs{p_{k+1}}-W_{k}(T)}{Z_{k+1}(T)Z_{k}(T)}
  2^{-\frac{\abs{p_{k+1}}}{T}}.
\end{equation*}
Recall here that,
for every $T\in(0,1)$,
$\lim_{k\to\infty}Z_{k}(T)$ and $\lim_{k\to\infty}W_{k}(T)$
exist and are positive by Lemma~\ref{uniform} (i).
It follows from $\lim_{k\to\infty}\abs{p_{k+1}}=\infty$ that,
for every $T\in(0,1)$,
there exists $k_0\in\N^+$ such that,
for every $k\ge k_0$, $E_{k}(T)<E_{k+1}(T)$
and therefore $S_{k}(T)<S_{k+1}(T)$ by \eqref{lnZ}.
Using Theorem~\ref{tdr} (iii),
we see that,
for every $T\in(0,1)$,
there exists $k_1\in\N^+$ such that,
for every $k\ge k_1$,
$0<S_{k_1}(T)\le S_{k}(T)<S(T)$.
Thus, using \eqref{diff-g-T},
for every $T,t\in(0,1)$ with $T<t$,
there exists $k_2\in\N^+$ such that
$S_{k_2}(T)>0$ and
for every $k\ge k_2$ and every $x\in(T,t)$,
$S_{k_2}(T)(x-T)\le g(x,k)-g(T,k)<S(t)(x-T)$.
Therefore,
the condition (iii) of Theorem~\ref{gfpwCTr}
and the condition (iii) of Theorem~\ref{gfpTc2}
hold for $g$.

Thus,
Theorem~\ref{gfpwCTr}, Theorem~\ref{gfpTc1}, and Theorem~\ref{gfpTc2}
result in the following three theorems, respectively.

\begin{theorem}\label{fpwcTr-free_energy}
For every $T\in(0,1)$,
if $F(T)$ is left-computable
then $T$ is weakly Chaitin $T$-random.\qed
\end{theorem}

\begin{theorem}\label{fpTc1-free_energy}
For every $T\in(0,1)$,
if $F(T)$ is left-computable
then $T$ is right-computable.\qed
\end{theorem}

\begin{theorem}\label{fpTc2-free_energy}
For every $T\in(0,1)$,
if both $F(T)$ and $T$ are right-computable
then $T$ is $T$-compressible.\qed
\end{theorem}

Theorem~\ref{fpt-free_energy} follows immediately from
these three theorems.

\subsection{The proof of Theorem~\ref{fpt-energy}}
\label{proof-fpt-energy}

We
complete the proof of
Theorem~\ref{fpt-energy},
based on Theorems~\ref{tdr}, \ref{gfpCTr}, \ref{gfpTc1}, and \ref{gfpTc2},
as follows.

Let $f\colon (0,1)\to\R$ with $f(T)=E(T)$,
and let $g\colon (0,1)\times\N^+\to\R$ with $g(T,k)=E_{k}(T)$.
First, by Theorem~\ref{tdr} (ii) and (iii),
we see that $E'(T)=C(T)>0$ for every $T\in(0,1)$.
Thus $f$ is increasing.

Obviously, $\lim_{k\to\infty} g(T,k)=f(T)$ for every $T\in(0,1)$.
At this point,
the conditions (i) and (ii) of Theorem~\ref{gfpCTr},
the conditions (i) and (iii) of Theorem~\ref{gfpTc1},
and the conditions (i), (v), and (vi) of Theorem~\ref{gfpTc2}
hold for $g$.

We see that,
for every $k\in\N^+$ and every $T\in(0,1)$,
\begin{equation}\label{dEk}
  E_{k+1}(T)-E_{k}(T)
  =
  \frac{Z_{k}(T)\abs{p_{k+1}}-W_{k}(T)}{Z_{k+1}(T)Z_{k}(T)}
  2^{-\frac{\abs{p_{k+1}}}{T}}.
\end{equation}
Recall here that,
for every $T\in(0,1)$,
$\lim_{k\to\infty}Z_{k}(T)$ and $\lim_{k\to\infty}W_{k}(T)$
exist and are positive by Lemma~\ref{uniform} (i).
It follows from $\lim_{k\to\infty}\abs{p_{k+1}}=\infty$ that,
for every $T\in(0,1)$,
there exist $a\in\N$, $b\in\N$, and $k_0\in\N^+$ such that,
every $k\ge k_0$,
\begin{equation}\label{dgEk}
  \abs{p_{k+1}}2^{-\frac{\abs{p_{k+1}}}{T}-a}
  \le
  g(T,k+1)-g(T,k)
  \le
  \abs{p_{k+1}}2^{-\frac{\abs{p_{k+1}}}{T}+b}.
\end{equation}
Thus,
the condition (iv) of Theorem~\ref{gfpCTr}
and the condition (iv) of Theorem~\ref{gfpTc2}
hold for $g$.
It follows from \eqref{dgEk} that,
for every $T\in(0,1)$,
there exists $k_0\in\N^+$ such that,
every $k\ge k_0$,
$g(T,k)<g(T,k+1)$ and therefore $g(T,k)<f(T)$.
Thus,
the condition (ii) of Theorem~\ref{gfpTc2} holds for $g$.

Using Lemma~\ref{uniform} (ii) and (iii)
in addition to Lemma~\ref{uniform} (i),
we can show a stronger statement than the inequalities \eqref{dgEk}.
The stronger statement for the lower bound of \eqref{dgEk} is
needed here.
That is,
based on \eqref{dEk},
Lemma~\ref{uniform},
and $\lim_{k\to\infty}\abs{p_{k+1}}=\infty$,
we can show that,
for every $T_1,T_2\in(0,1)$ with $T_1<T_2$,
there exist $a\in\N$ and $k_0\in\N^+$ such that,
every $k\ge k_0$ and every $x\in[T_1,T_2]$,
\begin{equation*}
  \abs{p_{k+1}}2^{-\frac{\abs{p_{k+1}}}{x}-a}
  \le
  g(x,k+1)-g(x,k).
\end{equation*}
It follows that
the condition (ii) of Theorem~\ref{gfpTc1} holds for $g$.

Now,
using the mean value theorem and Theorem~\ref{tdr} (i),
we see that,
for every $k\in\N^+$ and every $T,x,t\in(0,1)$ with $T<x<t$,
there exists $y\in(T,x)$ such that
\begin{equation}\label{dgET}
  g(x,k)-g(T,k)=C_{k}(y)(x-T).
\end{equation}
On the other hand,
using \eqref{CkT}
we see that,
for every $k\in\N^+$ and every $T\in(0,1)$,
$C_{k+1}(T)-C_{k}(T)$ is calculated as
\begin{equation*}
\begin{split}
  &\frac{\ln 2}{T^2}
  \frac{2^{-\frac{\abs{p_{k+1}}}{T}}}{Z_{k+1}(T)}
  \Biggl[
    \abs{p_{k+1}}^2
    -
    \left\{
      \frac{W_{k+1}(T)}{Z_{k+1}(T)}
      +
      \frac{W_{k}(T)}{Z_{k}(T)}
    \right\}
    \abs{p_{k+1}} \\
    &\hspace*{32mm}+
    \left\{
      \frac{W_{k+1}(T)}{Z_{k+1}(T)}
      +
      \frac{W_{k}(T)}{Z_{k}(T)}
    \right\}
    \frac{W_{k}(T)}{Z_{k}(T)}
    -
    \frac{Y_{k}(T)}{Z_{k}(T)}
  \Biggr].
\end{split}
\end{equation*}
Thus,
based on Lemma~\ref{uniform} and
$\lim_{k\to\infty}\abs{p_{k+1}}=\infty$,
we can show that,
for every $T_1,T_2\in(0,1)$ with $T_1<T_2$,
there exist $a\in\N$ and $k_0\in\N^+$ such that,
every $k\ge k_0$ and every $y\in[T_1,T_2]$,
\begin{equation*}
  \abs{p_{k+1}}^2 2^{-\frac{\abs{p_{k+1}}}{y}-a}
  \le
  C_{k+1}(y)-C_{k}(y).
\end{equation*}
It follows from Theorem~\ref{tdr} (iii) that,
for every $T_1,T_2\in(0,1)$ with $T_1<T_2$,
there exist $a\in\N$ and $k_0\in\N^+$ such that,
every $k\ge k_0$ and every $y\in[T_1,T_2]$,
\begin{equation}\label{mCM}
  0<\min C_{k_0}([T_1,T_2])\le C_{k}(y)<\max C([T_1,T_2]),
\end{equation}
where
$\min C_{k_0}([T_1,T_2])
=\min\{\,C_{k_0}(z)\mid z\in[T_1,T_2]\,\}$
and
$\max C([T_1,T_2])
=\max\{\,C(z)\mid z\in[T_1,T_2]\,\}$.
In particular,
$\max C([T_1,T_2])$
exists.
This is because
the function $C(T)$ of $T$ is continuous on $(0,1)$
by Lemma~\ref{uniform} and \eqref{CkT}. 
It follows from \eqref{dgET} and \eqref{mCM} that,
for every $T,t\in(0,1)$ with $T<t$,
there exist $a\in\N$, $b\in\N$, and $k_0\in\N^+$ such that,
for every $k\ge k_0$ and every $x\in(T,t)$,
$2^{-a}(x-T)\le g(x,k)-g(T,k)<2^{b}(x-T)$.
Therefore,
the condition (iii) of Theorem~\ref{gfpCTr}
and the condition (iii) of Theorem~\ref{gfpTc2}
hold for $g$.

Thus,
Theorem~\ref{gfpCTr}, Theorem~\ref{gfpTc1}, and Theorem~\ref{gfpTc2}
result in the following three theorems, respectively.

\begin{theorem}\label{fpcTr-energy}
For every $T\in(0,1)$,
if $E(T)$ is right-computable
then $T$ is Chaitin $T$-random.\qed
\end{theorem}

\begin{theorem}\label{fpTc1-energy}
For every $T\in(0,1)$,
if $E(T)$ is right-computable
then $T$ is also right-computable.\qed
\end{theorem}

\begin{theorem}\label{fpTc2-energy}
For every $T\in(0,1)$,
if $E(T)$ is left-computable and $T$ is right-computable,
then $T$ is $T$-compressible.\qed
\end{theorem}

Theorem~\ref{fpt-energy} follows immediately from
these three theorems.

\subsection{The proof of Theorem~\ref{fpt-entropy}}
\label{proof-fpt-entropy}

In a similar manner to the proof of Theorem~\ref{fpt-energy}
described in the previous subsection,
we can prove Theorem~\ref{fpt-entropy},
based on Theorems~\ref{tdr}, \ref{gfpCTr}, \ref{gfpTc1}, and \ref{gfpTc2}.
It is easy to convert the proof of Theorem~\ref{fpt-energy}
into the the proof of Theorem~\ref{fpt-entropy},
because of the similarity between
$E_{k}'(T)=C_{k}(T)$ and $S_{k}'(T)=C_{k}(T)/T$
given in Theorem~\ref{tdr} (i).

\section{Some properties of the sufficient conditions}
\label{psc}

In this section,
we investigate some properties of the sufficient conditions
for $T$ to be a fixed point
in the fixed point theorems on partial randomness;
Theorems~\ref{main}, \ref{fpt-free_energy}, \ref{fpt-energy},
and \ref{fpt-entropy}.
%
%
First we show
Theorems~\ref{scZ}, \ref{scF}, \ref{scE}, and \ref{scS} below
on the sufficient conditions.

\begin{theorem}[\cite{T08CiE}]\label{scZ}
The set
$\{\,T\in(0,1)\mid Z(T)\text{ is computable}\,\}$
is dense in $(0,1)$.
\end{theorem}

\begin{proof}
Since the function $2^{-l/T}$ of $T$ is
increasing on $(0,1)$ for each $l\in\N^+$,
the function $Z(T)$ of $T$ is increasing on $(0,1)$.
On the other hand,
the function $Z(T)$ of $T$ is continuous on $(0,1)$
by Lemma~\ref{uniform} (iii).
Thus,
since the set of all computable real numbers is dense in $\R$,
the result follows.
\end{proof}

\begin{theorem}\label{scF}
The set
$\{\,T\in(0,1)\mid F(T)\text{ is computable}\,\}$
is dense in $(0,1)$.
\end{theorem}

\begin{proof}
It follows from Theorem~\ref{tdr} (ii) and (iii) that
the function $F(T)$ of $T$ is
a decreasing continuous function on $(0,1)$.
Since
the set of all computable real numbers is dense in $\R$,
the result follows.
\end{proof}

In the same manner as the proof of Theorem~\ref{scF},
we can prove the following theorems for each of
$E(T)$ and $S(T)$.

\begin{theorem}\label{scE}
The set
$\{\,T\in(0,1)\mid E(T)\text{ is computable}\,\}$
is dense in $(0,1)$.\qed
\end{theorem}

\begin{theorem}\label{scS}
The set
$\{\,T\in(0,1)\mid S(T)\text{ is computable}\,\}$
is dense in $(0,1)$.\qed
\end{theorem}

Since the relation $F(T)=-T\log_2 Z(T)$ holds
from Definition~\ref{tdqait},
we can show the following theorem
on the relation between the sufficient conditions
in the fixed point theorems
by $Z(T)$ and $F(T)$.

\begin{theorem}
There does not exist $T\in(0,1)$ such that
both $Z(T)$ and $F(T)$ are computable.
\end{theorem}

\begin{proof}
Contrarily, assume that
both $Z(T)$ and $F(T)$ are computable
for some $T\in(0,1)$.
Since $F(T)=-T\log_2 Z(T)$ and $0<Z(T)<1$,
it is easy to see that $T$ is computable.
It follows from Theorem~\ref{cprpffe} (i) that
$Z(T)$ is weakly Chaitin $T$-random.
However, this contradicts the assumption that $Z(T)$ is computable,
and the result follows.
\end{proof}

Thus, the computability of $F(T)$ gives
completely different fixed points from the computability of $Z(T)$.
This implies that neither the computability of $Z(T)$ nor
the computability of $F(T)$ is the necessary condition
for $T$ to be a fixed point on partial randomness at all. 

In a similar manner,
we can
prove
the following two theorems
using the relations
$S(T)=E(T)/T+\log_2 Z(T)$ and $S(T)=(E(T)-F(T))/T$, respectively.

\begin{theorem}
There does not exist $T\in(0,1)$ such that
$Z(T)$, $E(T)$, and $S(T)$ are all computable.
\qed
\end{theorem}

\begin{theorem}
There does not exist $T\in(0,1)$ such that
$F(T)$, $E(T)$, and $S(T)$ are all computable.
\qed
\end{theorem}

%

Using the property of a fixed point in the fixed point theorems,
we can show the following theorem.

\begin{theorem}\label{pfp}
$S_a\cap S_b=\emptyset$
for any distinct computable real numbers $a,b\in(0,1]$,
where $S_a=\{\,T\in(0,1)\mid Z(aT)\text{ is computable}\,\}$.
\end{theorem}

\begin{proof}
Let $T\in(0,1)$, and let $a$ be a computable real number with $a\in(0,1]$.
Suppose that $Z(aT)$ is computable.
Then, by Theorem~\ref{main},
$\lim_{n\to\infty}H((aT)_n)/n=aT$.
Note that $H((aT)_n)=H(T_n)+O(1)$ for all $n\in\N^+$.
It follows that $\lim_{n\to\infty}H(T_n)/n=aT$.

Thus, for every computable real numbers $a,b\in(0,1]$,
if $S_a\cap S_b\neq\emptyset$ then $a=b$.
\end{proof}

As a corollary of Theorem~\ref{pfp},
we have the following, for example.

\begin{corollary}
For every $T\in(0,1)$,
if $Z(T)$ is computable,
then $Z(T/n)$ is not computable
for every $n\in\N^+$ with $n\ge 2$.
In other words,
for every $T\in(0,1)$,
if the
sum $\sum_{i=1}^{\infty} 2^{-\abs{p_i}/T}$ is computable,
then the corresponding power sum
$\sum_{i=1}^{\infty}\left(2^{-\abs{p_i}/T}\right)^n$ is not computable
for every $n\in\N^+$ with $n\ge 2$.
\qed
\end{corollary}

\section*{Acknowledgments}

This work was supported by
CREST of the Japan Science and Technology Agency,
by KAKENHI, Grant-in-Aid for Scientific Research (C) (20540134),
and by SCOPE
of the Ministry of Internal Affairs and Communications of Japan.

%


\appendix

\section{The proof of Theorem~\ref{main}}
\label{proof-main}

For completeness, we prove here Theorem~\ref{main},
based on Theorems~\ref{gfpwCTr}, \ref{gfpTc1}, and \ref{gfpTc2},
in a similar manner to
the proof of Theorem~\ref{fpt-free_energy}
given in Section~\ref{proof}.

Let $f\colon (0,1)\to\R$ with $f(T)=Z(T)$,
and let $g\colon (0,1)\times\N^+\to\R$ with $g(T,k)=Z_{k}(T)$.
First, it follows that $f$ is increasing.

Obviously, $\lim_{k\to\infty} g(T,k)=f(T)$ for every $T\in(0,1)$.
It follows that,
for every $T\in(0,1)$ and every $k\in\N^+$,
\begin{equation*}
  g(T,k+1)-g(T,k)=2^{-\frac{\abs{p_{k+1}}}{T}}.
\end{equation*}
Thus,
for every $T\in(0,1)$ and every $k\in\N^+$,
$g(T,k)<g(T,k+1)$ and therefore $g(T,k)<f(T)$.
At this point,
the conditions (i), (ii), and (iv) of Theorem~\ref{gfpwCTr},
all conditions of Theorem~\ref{gfpTc1},
and the conditions (i), (ii), (iv), (v), and (vi) of Theorem~\ref{gfpTc2}
hold for $g$.

Using the mean value theorem we see that,
for every $k\in\N^+$ and every $T,x,t\in(0,1)$ with $T<x<t$,
\begin{equation*}
  \frac{\ln 2}{t^2} W_{k}(T)(x-T)
  < g(x,k)-g(T,k)
  < \frac{\ln 2}{T^2}W_{k}(t)(x-T),
\end{equation*}
where
$W_{k}(T)=\sum_{i=1}^k \abs{p_i}2^{-\frac{\abs{p_i}}{T}}$,
as defined before Lemma~\ref{uniform} in Section~\ref{fpts}.
Note that,
for every $t\in(0,1)$,
$W_{k}(t)$ is increasing on $k$, and
$W(t)=\lim_{k\to\infty}W_{k}(t)$ exists by Lemma~\ref{uniform} (i).
Thus we see that
\begin{equation*}
  \frac{\ln 2}{t^2} W_{1}(T)(x-T)
  < g(x,k)-g(T,k)
  < \frac{\ln 2}{T^2}W(t)(x-T),
\end{equation*}
for every $k\in\N^+$ and every $T,x,t\in(0,1)$ with $T<x<t$.
Therefore,
the condition (iii) of Theorem~\ref{gfpwCTr}
and the condition (iii) of Theorem~\ref{gfpTc2}
hold for $g$.

Thus,
Theorem~\ref{gfpwCTr}, Theorem~\ref{gfpTc1}, and Theorem~\ref{gfpTc2}
result in the following three theorems, respectively.

\begin{theorem}\label{fpwcTr-partition_function}
For every $T\in(0,1)$,
if $Z(T)$ is right-computable
then $T$ is weakly Chaitin $T$-random.\qed
\end{theorem}

\begin{theorem}\label{fpTc1-partition_function}
For every $T\in(0,1)$,
if $Z(T)$ is right-computable
then $T$ is also right-computable.\qed
\end{theorem}

\begin{theorem}\label{fpTc2-partition_function}
For every $T\in(0,1)$,
if $Z(T)$ is left-computable and $T$ is right-computable,
then $T$ is $T$-compressible.\qed
\end{theorem}

Theorem~\ref{main} follows immediately from
these three theorems.


\begin{thebibliography}{99}
%


\bibitem{C85}
H. B. Callen,
\emph{Thermodynamics and an Introduction to Thermostatistics}, 2nd ed.
John Wiley \& Sons, Inc., Singapore, 1985.

\bibitem{CST06}
C. S. Calude, L. Staiger, and S. A. Terwijn,
``On partial randomness,''
\emph{Ann. Pure Appl. Logic},
vol.\ 138, pp.~20--30, 2006.

\bibitem{CS06}
C. S. Calude and M. A. Stay,
``Natural halting probabilities, partial randomness,
and zeta functions,''
\emph{Inform. and Comput.}, vol.\ 204, pp.~1718--1739, 2006.

\bibitem{C75}
G. J. Chaitin,
``A theory of program size formally identical to information theory,''
\emph{J. Assoc. Comput. Mach.}, vol.\ 22, pp.~329--340, 1975.

\bibitem{C87a}
G. J. Chaitin,
``Incompleteness theorems for random reals,''
\emph{Adv. in Appl. Math.}, vol.\ 8, pp.~119--146, 1987.

\bibitem{C87b}
G. J. Chaitin,
\emph{Algorithmic Information Theory}.
Cambridge University Press, Cambridge, 1987.




\bibitem{DH08}
R. G. Downey and D. R. Hirschfeldt,
\emph{Algorithmic randomness and complexity}.
Springer-Verlag, To appear.


\bibitem{G74}
P. G\'acs,
``On the symmetry of algorithmic information,''
\emph{Soviet Math. Dokl.}, vol.\ 15, pp.~1477--1480, 1974;
correction, ibid. vol.\ 15, pp.~1480, 1974.

\bibitem{K08a}
B. Kjos-Hanssen,
``Infinite subsets of random sets of integers,''
\emph{Math. Res. Lett.}, vol.\ 16, no.\ 1, pp.~103--110, 2009.

\bibitem{K08b}
B. Kjos-Hanssen,
Private communication, September 2008.

\bibitem{L74}
L. A. Levin,
``Laws of information conservation (non-growth) and
aspects of the foundations of probability theory,''
\emph{Problems of Inform. Transmission}, vol.\ 10, pp.~206--210, 1974.


\bibitem{Re65}
F. Reif,
\emph{Fundamentals of Statistical and Thermal Physics}.
McGraw-Hill, Inc., Singapore, 1965.

\bibitem{RS05}
J. Reimann and F. Stephan,
\newblock On hierarchies of randomness tests.
\newblock Proceedings of
the 9th Asian Logic Conference,
World Scientific Publishing,
August 16-19, 2005, Novosibirsk, Russia.

\bibitem{R08}
J. Reimann,
``Effectively closed sets of measures and randomness,''
To appear in \emph{Ann. Pure Appl. Logic}.

\bibitem{Ru99}
D. Ruelle,
\emph{Statistical Mechanics}, \emph{Rigorous Results}, 3rd ed.
Imperial College Press and World Scientific Publishing Co. Pte. Ltd.,
Singapore, 1999.


\bibitem{T99}
K. Tadaki,
\newblock Algorithmic information theory and fractal sets.
\newblock Proceedings of
1999 Workshop on Information-Based Induction Sciences (IBIS'99),
pp.~105--110,
August 26-27, 1999, Syuzenji, Shizuoka, Japan.
In Japanese.

\bibitem{T02}
K. Tadaki,
``A generalization of Chaitin's halting probability $\Omega$ and
halting self-similar sets,''
\emph{Hokkaido Math.\ J.}, vol.\ 31, pp.~219--253, 2002.
Electronic Version Available:
\url{http://arxiv.org/abs/nlin/0212001v1}

\bibitem{T07}
K. Tadaki,
\newblock A statistical mechanical interpretation of instantaneous codes.
\newblock Proceedings of
2007 IEEE International Symposium on Information Theory
(ISIT2007),
pp.~1906--1910,
June 24-29, 2007, Nice, France.

\bibitem{T08CiE}
K. Tadaki,
\newblock A statistical mechanical interpretation of
algorithmic information theory.
\newblock
Local Proceedings of Computability in Europe 2008 (CiE 2008),
pp.~425--434,
June 15-20, 2008, University of Athens, Greece.
Extended and Electronic Version Available:
\url{http://arxiv.org/abs/0801.4194v1}

\bibitem{T08ISIT}
K. Tadaki,
\newblock The Tsallis entropy and the Shannon entropy of
a universal probability.
\newblock Proceedings of
the 2008 IEEE International Symposium on Information Theory
(ISIT 2008),
pp.~2111--2115,
July 6-11, 2008, Toronto, Canada.

\bibitem{TKS92}
M. Toda, R. Kubo, and N. Sait\^o,
\emph{Statistical Physics} I.
\emph{Equilibrium Statistical Mechanics}, 2nd ed.
Springer, Berlin, 1992.

\bibitem{W00}
K. Weihrauch,
\emph{Computable Analysis}.
Springer-Verlag, Berlin, 2000.

\end{thebibliography}
\end{document}